\documentclass[10pt,journal]{IEEEtran}
\usepackage{amsthm}
\usepackage{times}
\usepackage{algorithm}
\usepackage[noend]{algorithmic}
\usepackage{algorithmic}
\usepackage{amsmath,graphicx,subfigure}
\usepackage{url,citesort}
\usepackage{caption}

\newtheorem{theorem}{Theorem}

\title{Joint Caching, Routing, and Channel Assignment for Collaborative Small-Cell Cellular Networks}
\author{\IEEEauthorblockN{Abdallah Khreishah$^\star$, Jacob Chakareski$^\dag$, and Ammar Gharaibeh$^\star$\\}
\IEEEauthorblockA{{$^\star$}New Jersey Institute of Technology; {$^\dag$}The University of Alabama, Tuscaloosa;}

\thanks{The work of Jacob Chakareski was partially supported by the NSF award CCF-1528030}
}

\begin{document}
\maketitle

\begin{abstract}
We consider joint caching, routing, and channel assignment for video delivery over coordinated small-cell cellular systems of the future Internet. We formulate the problem of maximizing the throughput of the system as a linear program in which the number of variables is very large. To address channel interference, our formulation incorporates the conflict graph that arises when wireless links interfere with each other due to simultaneous transmission. We utilize the column generation method to solve the problem by breaking it into a restricted master subproblem that involves a select subset of variables and a collection of pricing subproblems that select the new variable to be introduced into the restricted master problem, if that leads to a better objective function value. To control the complexity of the column generation optimization further, due to the exponential number of independent sets that arise from the conflict graph, we introduce an approximation algorithm that computes a solution that is within $\epsilon$ to optimality, at much lower complexity. Our framework demonstrates considerable gains in average transmission rate at which the video data can be delivered to the users, over the state-of-the-art Femtocaching system, of up to 46\%. These operational gains in system performance map to analogous gains in video application quality, thereby enhancing the user experience considerably.
\end{abstract}


\section{Introduction}\label{sec:Introduction}
Cellular networks are increasingly serving as our primary Internet-access facility. Simultaneously,
our unceasing appetite for online video and the omnipresence of social networking applications have led to
video data traffic consuming the bulk of their bandwidth \cite{Cisco:14nourl}. Therefore, new traffic engineering approaches are needed that will stem the online video data deluge and enable more efficient resource
utilization. Caching or video traffic offloading in small-cell wireless systems has emerged as one such
recent approach. In particular, instead of forwarding the user requests for video content to the macro-cell
base station, over expensive and bandwidth-limited back-haul links, they are served locally from the
small-cell base station that maintains a cache of the most-popular video content. The main question to be
answered in this context is how to populate the caches of every small-cell, given their capacities and the video file popularity distribution in every small cell, knowing that some users may be served from multiple caches, as they move around. A typical caching scenario of this nature is illustrated in Figure~\ref{fig:SmallCellScenario}.

\begin{figure}[htb]
\centering
\includegraphics[width=0.8\columnwidth]{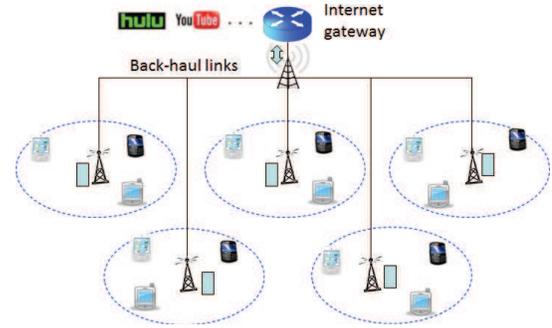}
\caption{Heterogeneous network caching: A collection of small-cells embedded into a macro-cell are served by small base-stations with caches installed. In case of a cache miss, the user requests are relayed to the macro-cell base station over back-haul links that forwards them to a remote Internet server.}
\label{fig:SmallCellScenario}
\end{figure}

Studies to date have examined the above setting from different perspectives. For example, in \cite{golrezaei2012femtocaching}, the authors formulated an information-theoretic objective of minimizing
the delivery cost of the requested video content in small-cell wireless networks with caching. In \cite{blasco2014learning}, the authors employed the multi-armed bandit
formulation to study the statistical risk of learning the content popularity distribution at a small-cell cache.
Three major shortcomings of the current state-of-the-art motivate our study. First, {\em collaboration} between
small-cell caches is not considered, i.e., the typical system architecture that is proposed is hierarchical in
nature, where in the event of a local cache miss, the user is redirected to the Internet via expensive and bandwidth-limited back-haul links\footnote{Frequently, these links are wireless in nature, and their use to relay users to the Internet via the macro-cell base station reduces the aggregate transmission capacity available in the cell. Thus, they should be used sparingly.}, illustrated in Figure~\ref{fig:SmallCellScenario}. As shown in our recent preliminary study \cite{khreishah2015collaborative}, cooperation between different small-cell caches in serving their respective user populations can dramatically reduce the back-haul cost of serving the requested content from a remote Internet server. Second, the issue of {\em interference} is typically considered orthogonal to the caching problem under consideration, i.e., it is assumed that back-haul communication and cache-to-user communication do not interfere with other links transmitting at the same time. Third, {\em cache misses} at different small-cells are handled uniformly, which does not account for the unequal impact that misses at different small-cell caches will have on the overall performance. Consider for instance the setup illustrated in Figure~\ref{fig:model}. A cache miss will require transmission from the macro-cell base station to the requesting user. If the user is located in a nearby small-cell, then the communication channel over which the content is delivered can be reused at other small-cells located further away from the macro-cell base station. However, if the user is located in a remote small-cell, the transmission from the macro-cell base station will effectively block this channel from being reused across the whole macro-cell. Thus, cache misses should be handled {\em differentially}, depending on their location relative to the macro-cell base station.

We address the above challenges by formulating a framework for joint caching, interference management, and routing in coordinated small-cell wireless systems. Our framework advances the state-of-the-art considerably, by dispensing with unrealistic system assumptions, introducing novel small-cell collaboration concepts, and augmenting the system's serving capacity. As observed from our experiments, such advances can lead to considerable gains in achieved user throughput that in turn will map to enhanced user experience in terms of delivered video application quality. In the remainder of this section, we review related work in greater detail.

Caching at the cellular level has been studied extensively. The work in \cite{erman2011cache} develops a simple cost model that cellular network operators can use to determine the benefits of caching at a base station. The authors in \cite{ahlehagh2012video} study video caching in base stations to improve the users' quality of experience (QoE), among other factors. The work in \cite{khreishah2015collaborative} considers collaborative caching between base stations in different cells to minimize the aggregate cost of file delivery. However, all these studies only consider caching at the macro-cell level, where back-haul links can be used to carry the inter-base-station traffic, and no small-cells are considered. It should be emphasized that back-haul links are expensive to install and maintain, and their use leads to reduction in aggregate transmission capacity across a macro-cell, if they are wireless in nature. Thus, there is a system-efficiency cost associated with their use, as well.

The work in \cite{blasco2014learning} only considers the caching assignment problem and aims to minimize the total delay of file retrieval by estimating the file popularity distribution at a single small-cell base station. The authors in \cite{bastug2014cache} tackle the caching assignment problem in small-cell base stations by characterizing the considered performance metrics (outage probability and average data delivery rate) using system parameters (location of the macro base station, file popularity, etc.). The work in \cite{ostovari2013cache} investigates network coding and small-cell caching for enhanced network performance. In \cite{karamchandani2014hierarchical}, the authors formulate an information-theoretic objective of minimizing the delivery cost of the requested video file in small-cell wireless networks with caching. These studies only focus on the caching assignment problem and do not consider the routing and channel assignment problems that arise in such environments.

The most related studies to ours are \cite{golrezaei2012femtocaching,poularakis2014approximation,jain2005impact}. The work in \cite{jain2005impact} considers routing and channel assignment between nodes in a wireless network by using conflict graphs in order to maximize the network throughput. However, this work does not consider caching. The work in \cite{golrezaei2012femtocaching} was the first to propose caching at small-cell base stations. The authors exploit caching at small-cell base stations and perform link scheduling in order to maximize the number of satisfied users. However, caching and link scheduling are considered separately. The work in \cite{poularakis2014approximation} also considers caching at small-cell base stations by taking into account the transmission bandwidth assigned to every small-cell. However, this study does not account for the interference issues that arise when the users are served the requested content. As stated earlier, the present paper is the first to consider jointly caching, routing, and channel assignment in collaborative small-cell cellular networks.

The rest of the paper is organized as follows. Next, we introduce our system architecture and models we use. The problem formulation is presented in Section~\ref{sec:formulation}. The $(1\pm\epsilon)$ approximation algorithm, based on column generation, that computes the optimal solution is presented in Section~\ref{sec:column}. Various implementation aspects of our system are discussed in Section~\ref{sec:practical}. We carry out a complexity analysis of the optimization in Section~\ref{sec:performance}. Performance analysis of our system is carried out in Section~\ref{sec:simulation}. Finally, concluding remarks are provided in Section~\ref{sec:conclusion}.

\section{System Model} \label{sec:model}
\subsection{Network Model}
We consider a single macro-cell of a cellular network, comprising a macro base station (MBS), located at the center of the cell, and $\mathcal{N} = \{1, 2, \ldots, n, \ldots, N\}$ small base stations (SBS), scattered across the macro-cell, as illustrated in Figure~\ref{fig:model}. Let $\mathcal{N} = \mathcal{N} \cup \{MBS\}$. The $n$-th SBS is equipped with storage capacity of size $Cap_n$ and $a_n$ antennae that can be used for communication. We assume that the MBS has enough storage capacity to store all files, so that a requested file can always be served. There are $\mathcal{J} = \{1, 2, \ldots, j, \ldots, J\}$ files that can be requested by $\mathcal{K} = \{1, 2, \ldots, k, \ldots, K\}$ users, where the $k$-th user has $a_k$ antennae that can be used for communication. Let $\alpha_{kj}$ denote the number of requests per time slot generated by the $k$-th user requesting the $j$-th file. Let $S_j$ denote the size of the $j$-th file. The base stations and the users can communicate using a set $\mathcal{C} = \{1, 2, \ldots, c, \ldots C\}$ of available secondary channels with different bandwidths, where each SBS can use a subset of the secondary channels. The macro BS can use all of the secondary channels in addition to the basic channel denoted by $\{0\}$. Thus, let $\mathcal{C} = \mathcal{C} \cup \{0\}$. We assume a single hop communication where the $k$-th user can retrieve a file from the $n$-th SBS only if the $k$-th user is within the transmission range of the $n$-th SBS. Since we only consider the downlink session, we will use the terms SBS and transmitter interchangeably, as well as the terms user and receiver. 

\begin{figure}
\centering
\includegraphics[width=\columnwidth]{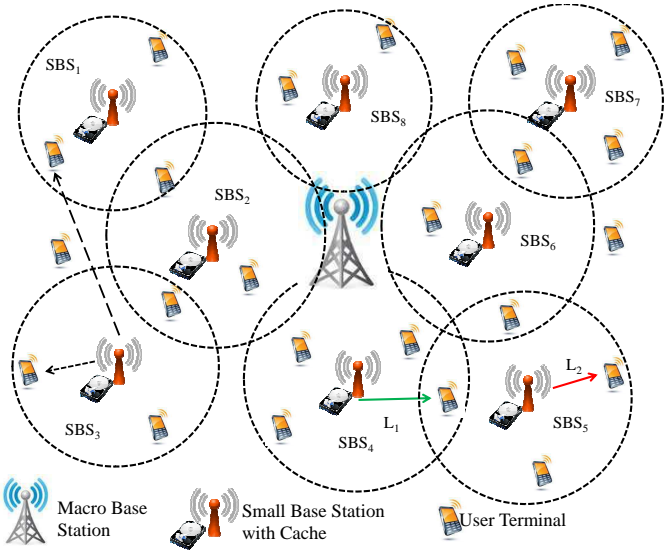}
\caption{Heterogeneous macro-cell and small-cell cellular network model.}
\label{fig:model}
\end{figure}

\subsection{Network Coding}
Network coding improves throughput and lowers scheduling complexity, by enabling efficient packet transmission and polynomial-time optimization solutions \cite{ho2008network}. We consider intra-session network coding has been applied to the data packets of each video file. In particular, if a file comprises packets $\{q_1, q_2,\ldots,q_M\}$, we generate coded packets $\mathcal{Q} = \sum_{m = 1}^{M}\kappa_m q_m$, where $\kappa_m, \forall m$ are random coefficients and the arithmetic operations are performed over a finite field \cite{ho2006random}. This will ensure that any $M$ coded packets will be linearly independent with very high probability, and the user can reconstruct the file by receiving $M$ different coded packets.

\subsection{Transmission/Interference Range and Link Capacity}
We assume that the transmission and interference ranges of each base station, when transmitting on a certain channel, are fixed. This can be done by fixing the transmission power $P_n^c$ of the $n$-th SBS when communicating using the $c$-th channel. We use the Protocol Interference Model \cite{gupta2000capacity}, where the transmission from the $n$-th SBS to the $k$-th user using the $c$-th channel is successful if the received power is higher than a threshold $P_T^c$, and the interfernece is significant if the received power is higher than a threshold $P_I^c$. Therefore, the transmission range ${TR}_{n}^{c}$ and interference range ${IR}_{n}^{c}$ of the $n$-th SBS can be calculated as follows:
\begin{displaymath}
{TR}_{n}^{c} = (g P_n^c/P_T^c)^{1/\gamma}, \quad {IR}_{n}^{c} = (g P_n^c/P_I^c)^{1/\gamma}
\end{displaymath}
where $\gamma$ is the path loss factor and $g$ is a constant related to the wavelength, the antenna profiles at the transmitter and receiver, and other factors.

The capacity $\hat{c}_{nk}^{c}$ of the link between the $n$-th SBS and the $k$-th user when communicating using the $c$-th channel is calculated according to the Shannon-Hartley theorem as follows
\begin{displaymath}
\hat{c}_{nk}^{c} = W_c \log_2(1 + \frac{G_{nk}P_n^c}{\eta})
\end{displaymath}
where $W_c$ is the bandwidth of the $c$-th channel, $\eta$ is the white noise at the receiver, and $G_{nk} = g\cdot(d_{nk})^{-\gamma}$ is the power propagation gain between the $n$-th SBS and the $k$-th user when the Euclidean distance is $d_{nk}$. Note that due to the Protocol Interference Model, an SBS cannot transmit to a user on a channel if the user falls inside the interference range of another SBS using the same channel. Thus the capacity of the link is calculated using Signal-to-Noise ratio (SNR) instead of Signal to Interference plus Noise ratio (SINR).

\subsection{Conflict Graph and Independent Sets}
In order to characterize the interference among the communication links in our system, we construct a conflict graph. A conflict graph $G(V, E)$ is a graph where every vertex in the vertex set $V$ represents a communication link tuple $((n,k),c)$, where $n \in \mathcal{N}, k \in \mathcal{K}, c \in \mathcal{C}$. An edge connecting two vertices in the conflict graph means that the corresponding communication links cannot be active at the same time. This happens if the receiver in one tuple falls within the interference range of the transmitter in the other tuple when communicating on the same channel. A transmitter uses the same channel to send to multiple receivers, or a receiver receives multiple signals on the same channel.

An independent set $I$ is a set containing communication tuples that do not interfere with each other (i.e., there is no edge connecting any two corresponding verices in the conflict graph). Since all the communication links in an independent set do not interfer with each other, these links can be active and used for communication simultaneously. If adding another link to the independent set $I$ makes it non-independent, then $I$ is defined as a maximal independent set. We represent the independent set $I$ by a $0-1$ vector of length equal to the number of communication tuples, where a value of $1$ in $I$ indicates that the corresponding communication tuple is included in $I$.

Figure~\ref{fig:IS} illustrates an example on how to construct the conflict graph and the independent sets. In this example, we have two SBSs and three users. SBS 1 and user 2 can communicate using two channels $1$ and $2$, while all other SBSs and users can communicate using channel $1$ only. The arrows indicate the available communication tuples. These tuples are represented by the vertices in the conflict graph. An edge connecting two vertices means that the corresponding communication tuples cannot be active at the same time. For example, since there is an edge connecting $v_1$ and $v_4$, then the tuples $((1,1),1)$ and $((2,2),1)$ cannot be active at the same time, because when SBS 1 transmits to user 1 on channel $1$, user 2 will suffer from interference when recieving on channel $1$ from SBS 2. From this conflict graph, we can construct the desired independent sets. For example, $I_1 = (1,0,0,0,1)$ is an independent set and $I_2 = (1,0,1,0,1)$ is a maximal independent set since adding any new vertex to $I_2$ will make it non-independent.

\begin{figure}
\centering
\includegraphics[width=\columnwidth]{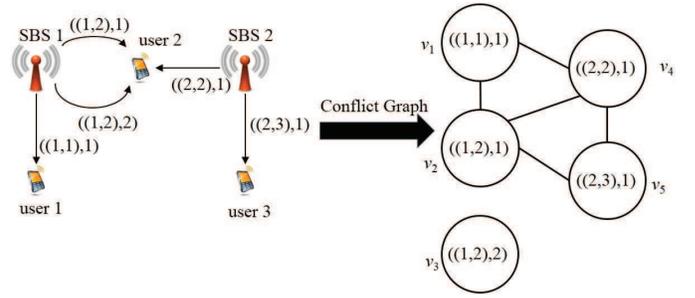}
\caption{Conflict graph construction example.}
\label{fig:IS}
\end{figure}

\section{Problem Formulation} \label{sec:formulation}
In this section, we formulate the problem of jointly caching, routing, and scheduling for cellular small cell networks.

\subsection{Formulation}
Before we state our problem formulation, we introduce the following variables:
\begin{itemize}
\item $X_{nj}$: The fraction of the $j$-th file stored at the $n$-th SBS.

\item $Y_{nj}^k$: The fraction of the $j$-th file requested by the $k$-th user and served from the $n$-th SBS.

\item $Z_{nk}^c$: Total data transfered from the $n$-th SBS to the $k$-th user using the $c$-th channel.

\item $f_i$: The percentage of time the $i$-th independent set $I_i$ is active.

\item $\mathcal{I}$: The set of all independent sets
\end{itemize}
The problem formulation of jointly caching, routing, and scheduling for cellular small-cell networks is formulated using the following linear program (LP):
\begin{displaymath} \delta^* = \min \sum_{1 \leq i \leq \left\vert{\mathcal{I}}\right\vert} f_i \label{eqn:basic_formulation}
\end{displaymath}

\begin{align}
&\text{s.t.~} \sum_{j}S_{j}X_{nj} \leq Cap_{n} \quad &\forall n \label{const1}\\
& \sum_{n} Y_{nj}^{k} \geq 1_{\{\alpha_{kj} > 0\}} \quad &\forall k, j \label{const2}\\
& Y_{nj}^{k} \leq X_{nj} \quad &\forall n, k, j \label{const3}\\
& \sum_{j} Y_{nj}^{k}\alpha_{kj}S_{j} \leq \sum_{c}Z_{nk}^{c} \quad &\forall n, k \label{const4}\\
& Z_{nk}^{c} \leq \sum_{1 \leq i \leq \left\vert{\mathcal{I}}\right\vert}f_i\hat{c}_{nk}^{c}(I_i) \quad &\forall n, k, c \label{const5}\\
& \sum_{1 \leq i \leq \left\vert{\mathcal{I}}\right\vert} f_i \leq 1 \label{const6}
\end{align}

In this problem, the objective is to minimize the schedule length required to satisfy all users\footnote{That is, the amount of time required to serve all user requests.}. Constraint \eqref{const1} guarantees the cache capacity constraints of the $n$-th SBS. Constraint \eqref{const2} states that the sum of the fractions of a file received by a user is at least 1. Constraint \eqref{const3} states that a user can retrieve a file from an SBS only if that SBS has the file in its cache. Constraints \eqref{const4} and \eqref{const5} together state that the total size of data transmitted from the $n$-th SBS to the $k$-th user on the $c$-th channel should be less or equal to the capacity $\hat{c}_{nk}^{c}$ of the communication tuple $((n,k),c)$ times the fraction of time the independent sets containing the tuple $((n,k),c)$ are active. The last constraint states that a feasible schedule length should be less or equal to 1.

The above formulation can be adjusted to include other objectives. If the objective is to maximize the total throughput, then the objective function is changed to
\begin{displaymath}
\max \sum_n \sum_k \sum_c Z_{nk}^{c} \, .
\end{displaymath}

\section{An $\epsilon$-Bounded Approximation Algorithm Based on Column Generation} \label{sec:column}
The challenge posed by the problem formulation presented in Section \ref{sec:formulation} is that we may have a large number of independent sets, based on the conflict graph, and therefore, a large number of variables to deal with. However, most of these variables will not contribute to the objective function as their value will be zero \cite{bertsimas1997introduction}. Column generation takes advantage of this observation and only selects the variables that have the potential to improve the objective function. Based on this, we divide the original optimization problem into two subproblems called the Restricted Master problem (RMP) and the Pricing problem (PP). The RMP starts by working on an initial subset of variables. The PP uses the optimal dual solution to the RMP to identify a new variable with the most negative reduced cost, relative to the objective function of the RMP. This variable is then added back to the RMP, and the process is repeated until an $\epsilon$-bounded solution is obtained. The process of column generation is shown in Figure~\ref{fig:CG}.

\begin{figure}[h]
\centering
\includegraphics[width=0.8\columnwidth]{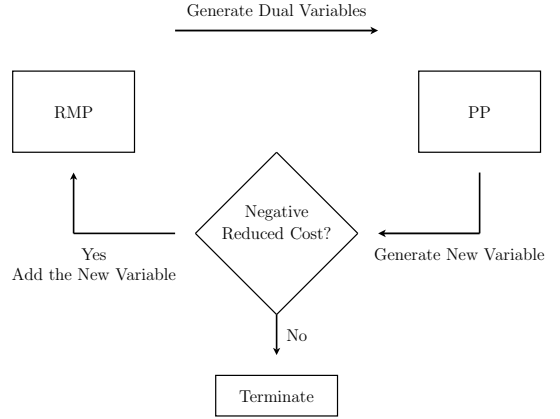}
\caption{The column generation process.}
\label{fig:CG}
\end{figure}

\subsection{The Restricted Master Problem}
As stated above, the RMP starts by considering an initial subset $\mathcal{I}' \subset \mathcal{I}$ of independent sets, which can be obtained by including a single communication tuple in each independent set $I\in \mathcal{I}'$ (i.e. $\mathcal{I}'$ is similar to the identity matrix). Therefore, the formulation of the RMP is adjusted as follows:
\begin{displaymath} \delta = \min \sum_{1 \leq i \leq \left\vert{\mathcal{I}'}\right\vert} f_i \label{eqn:RMP_basic_formulation}
\end{displaymath}
subject to \eqref{const1}, \eqref{const2}, \eqref{const3}, \eqref{const4} and
\begin{align}
& Z_{nk}^{c} \leq \sum_{1 \leq i \leq \left\vert{\mathcal{I}'}\right\vert}f_i\hat{c}_{nk}^{c}(I_i) \quad \forall n, k, c \label{const7}\\
& \sum_{1 \leq i \leq \left\vert{\mathcal{I}'}\right\vert} f_i \leq 1 \label{const8}
\end{align}

Since RMP is a small-scale linear program, it can be solved to optimality in polynomial time \cite{bertsimas1997introduction}, and we can obtain its dual optimal solution. However, since RMP considers a subset $\mathcal{I}'$ of all independent sets, the optimal solution to RMP is an upper bound on the optimal solution of the original problem. This upper bound is decreased when more independent sets are included in the RMP. Thus, we use the PP to determine which independent set, which will introduce a new variable $f_i$, has the potential to reduce the objective function of the RMP the most. This process continues until we get close enough to the optimal solution of the original problem.

\subsection{The Pricing Problem}
The goal of the PP is to generate an independent set $I_i \in \mathcal{I}/\mathcal{I}'$ that can reduce the cost of the objective function of the RMP. The reduced cost of an independent set $I_i \in \mathcal{I}/\mathcal{I}'$ is calculated as \cite{bertsimas1997introduction}:
\begin{displaymath}
\omega_i = 1 - \sum_{n,k,c} \lambda_{nk}^{c}\hat{c}_{nk}^{c}t_{nk}^{c}
\end{displaymath}
where $\lambda_{nk}^{c}$ are the dual variables corresponding to \eqref{const7}, and $t_{nk}^{c}$ is a binary variable indicating whether the communication tuple $((n,k),c)$ is included in the $i$-th independent set or not. Since we need to find the independent set with the most negative reduced cost, the objective function of PP is:
\begin{equation}
\min_{I_i \in \mathcal{I}/\mathcal{I}'} \omega_i \label{eqn:PP_basic_formulation1}
\end{equation}
or equivalently
\begin{equation}
\max_{I_i \in \mathcal{I}/\mathcal{I}'} \beta_i = \sum_{n,k,c} \lambda_{nk}^{c}\hat{c}_{nk}^{c}t_{nk}^{c} \label{eqn:PP_basic_formulation2}
\end{equation}
Let $u_{i}^{*}$ and $\beta_{i}^{*}$ denote the optimal solution of \eqref{eqn:PP_basic_formulation1} and \eqref{eqn:PP_basic_formulation2} respectively. Since we are seeking the most negative reduced cost, the process of column generation terminates when $u_{i}^{*} \geq 0$ or $\beta_{i}^{*} \leq 1$.

The generated independent set using PP must be a feasible set (i.e. there is no interference between any two communication tuples included in the independent set). To ensure this, the following set of constraints are imposed on PP
\begin{align}
&\sum_k t_{nk}^{c} \leq 1 \quad &\forall n,c \label{const9}\\
&\sum_n t_{nk}^{c} \leq 1 \quad &\forall k,c \label{const10}\\
&\sum_{k,c} t_{nk}^{c} \leq a_n \quad &\forall n \label{const11}\\
&\sum_{n,c} t_{nk}^{c} \leq a_k \quad &\forall k \label{const12}\\
&t_{nk}^{c} + \sum_{\substack{n'\neq n\vert k\in \mathcal{F}_{n'}\\ k'\neq k}}t_{n'k'}^{c} \leq 1 \quad &\forall n,k,c \label{const13}\\
&t_{nk}^{c} \in \{0, 1\} \quad &\forall n,k,c \label{const14}
\end{align}
where \eqref{const9} states that a transmitter cannot send on more than one link using the same channel, \eqref{const10} states that a receiver cannot receive on more than one link using the same channel, \eqref{const11} states that the number of links used by a transmitter should be less or equal than the number of antennae it has, \eqref{const12} states that the number of links used by a receiver should be less or equal to the number of antennae it has, \eqref{const13} states that if the tuple $((n,k),c)$ is active, then any other tuple $((n',k'),c)$ where $k$ is in the interference range of the $n'$-th transmitter (i.e. $k \in \mathcal{F}_{n'}$, where $\mathcal{F}_{n'}$ is the set of receivers that fall in the interference range of $n'$) cannot be active at the same time, as this will introduce interference at the $k$-th receiver, and \eqref{const14} indicates the binary nature of the variable $t_{nk}^{c}$.

\subsection{$\epsilon$-Bounded Algorithm}
As mentioned before, the number of independent sets can be very large. In fact, the number of independent sets is exponential in terms of the number of communication links in the network. Therefore, the number of iterations needed to find all independent sets that has a negative reduced cost may be very large. However, it has been observed that a very close solution to the optimal solution can be found quickly \cite{lasdon2013optimization}. We take advantage of this observation to introduce an $\epsilon$-bounded approximation algorithm to find an $\epsilon$-bounded solution. A solution $\delta$ to the original problem is said to be an $\epsilon$-bounded solution if it satisfies $(1-\epsilon)\delta^* \leq \delta \leq (1+\epsilon)\delta^*$, where $\delta^*$ is the optimal solution to the original problem stated in \eqref{eqn:basic_formulation} in Section~\ref{sec:formulation}. The $\epsilon$-bounded approximation algorithm is introduced in Algorithm~\ref{alg:alg1}. Based on this definition, we introduce the following theorem

\begin{theorem}
Let $\delta^u, \delta^l$ computed in Algorithm~\ref{alg:alg1} denote the upper bound and the lower bound solutions on the optimal solution $\delta^*$ of the original problem. When the algorithm terminates, an $\epsilon$-bounded solution is obtained
\end{theorem}

\begin{proof}
The algorithm terminates either when $\beta_i \leq 1$, in which case, there is no independent set that has a negative reduced cost, and thus we have reached the optimal solution, or when $\frac{\delta^l}{\delta^u} \geq \frac{1}{1+\epsilon}$, in which case we have $\delta \leq \delta^u \leq (1+\epsilon)\delta^l \leq (1+\epsilon)\delta^*$ and $\delta \geq \delta^l \geq (1-\epsilon)\delta^u \geq (1-\epsilon)\delta^*$. Therefore, $\delta$ is an $\epsilon$-bounded solution.
\end{proof}

To obtain $\delta^u$ and $\delta^l$, note that the solution to the RMP is an upper bound to the solution of the original problem. Therefore, $\delta^u$ can be set as the solution to RMP. On the other hand, the lower bound $\delta^l$ can be computed as \cite{bertsimas1997introduction}
\begin{displaymath}
\delta^l = \max \{\delta^u + \Phi\omega_i^*, 0\}
\end{displaymath}
where $\omega_i^*$ is the optimal solution to the PP, and $\Phi \geq \sum_{1 \leq i \leq \left\vert{\mathcal{I}'}\right\vert} f_i$ holds for the optimal solution to the RMP \cite{bertsimas1997introduction}. Since a feasible solution must satisfy $\sum_{1 \leq i \leq \left\vert{\mathcal{I}'}\right\vert} f_i \leq 1$, we set $\Phi = 1$. Lastly, a feasible solution should be greater than zero.

\begin{algorithm}
\caption{$\epsilon$-Bounded Approximation Algorithm}
\label{alg:alg1}
\begin{algorithmic}
\STATE{Input: Approximation factor $\epsilon$, Initial subset of independent sets $\mathcal{I}'$, $\delta^u = \infty$, $\delta^l = 0$}
\STATE{Output: $\delta^u, \delta^l, f_i^{*}, X_{nj}, Y_{nj}^{k}$}
\WHILE{$\frac{\delta^l}{\delta^u} < \frac{1}{1 + \epsilon}$ and $\beta_{i}^{*} > 1$}
\STATE{Solve RMP to obtain its optimal solution $\delta^u$ and the dual optimal solution $\lambda_{nk}^{c}$}
\STATE{Using $\lambda_{nk}^{c}$, solve PP to generate an independent set $I_i$ and obtain $\beta_i^*$}
\STATE{Update $\mathcal{I}' = \mathcal{I}' \cup I_i$}
\STATE{$\omega_i^* = 1 - \beta_i^*$}
\STATE{}
\STATE{$\delta^l = \max\{\delta^u + \Phi\omega_i^*, 0\}$}
\ENDWHILE
\IF{$\delta^u \leq 1$}
\STATE{Demand can be supported}
\STATE{Cache the files according to $X_{nj}$}
\STATE{Route the files according to $Y_{nj}^{k}$}
\ELSIF{$\delta^u > 1+\epsilon$ or $\delta^l > 1$}
\STATE{Demand cannot be supported}
\ELSE
\STATE{set $\epsilon = 0$ to see if demand can be supported}
\ENDIF
\end{algorithmic}
\end{algorithm}

\section{Practical Issues}\label{sec:practical}
In this section, we discuss various practical aspects related to the real-life implementation of our algorithm.
\subsection{The Caching Process}
One of the outputs of our algorithm is where to store each file. If the algorithm is executed frequently, then changes in the file popularities in successive executions might change the cache assignment. This will cause data to be moved around the network, which is expensive due to the limited backhaul links. One possible solution is to execute the algorithm on a daily basis, where the cache assignment can be carried out during off-peak hours (e.g., during the night). The routing and link scheduling can then be performed when the actual requests start appearing.

\subsection{File Popularity Estimation}\label{subsec:estimation}
The algorithm assumes perfect knowledge of the video file user requests and the video file popularity distribution at every small-cell cache. Based on this, the algorithm outputs the caching, routing, and link scheduling assignments. However, this knowledge is not available beforehand and thus file popularity estimation has to be performed. If the estimation were correct, then we could schedule each independent set on a round-robin basis, such that the fraction of time the $i$-th independent set is scheduled is equal to $f_i$. However, if the estimation is wrong, we can rerun the algorithm based on a fixed caching assignment to obtain routing and link scheduling assignments. In this case, the $X$ variables in the original formulation will become constant, and the formulation to be solved will be the following:

\begin{displaymath} \delta^* = \min \sum_{1 \leq i \leq \left\vert{\mathcal{I}}\right\vert} f_i \label{eqn:basic_formulation_file_popularity_estimation}
\end{displaymath}

\begin{align}
&\text{s.t.~} \sum_{n} Y_{nj}^{k} \geq 1_{\{\alpha_{kj} > 0\}} \quad &\forall k, j \label{const15}\\
& Y_{nj}^{k} \leq X_{nj} \quad &\forall n, k, j \label{const16}\\
& \sum_{j} Y_{nj}^{k}\alpha_{kj}S_{j} \leq \sum_{c}Z_{nk}^{c} \quad &\forall n, k \label{const17}\\
& Z_{nk}^{c} \leq \sum_{1 \leq i \leq \left\vert{\mathcal{I}}\right\vert}f_i\hat{c}_{nk}^{c}(I_i) \quad &\forall n, k, c \label{const18}\\
& \sum_{1 \leq i \leq \left\vert{\mathcal{I}}\right\vert} f_i \leq 1 \label{const19}
\end{align}

Moreover, the observed user request patterns can be stored along with the algorithm's output to avoid rerunning the algorithm when the same pattern is observed again.

\subsection {Algorithm Implementation}
The execution of Algorithm~\ref{alg:alg1} is carried out by the Mobility Management Entity (MME) of the cellular network, which acts as a centralized controller \cite{3gpp.36.300}. Different architectures for cellular networks proposed in the literature, e.g., Software Defined Networks (SDN) \cite{li2012toward}, MobileFlow \cite{pentikousis2013mobileflow}, MOCA \cite{banerjee2013moca}, and SoftCell \cite{jin2013softcell}, have taken advantage of this centralized controller. Our algorithm is executed by the MME that has access to all the necessary information such as the network topology, the users' locations, and the files cached at each small base station. Based on this information, the MME runs the algorithm and decides at which small base station to cache a certain file, which base station is used to satisfy a user's demand, and which links are to be active at the same time (i.e., which independent set should be active).

\subsection{File Popularity Dynamics}
The file popularity at each small base station can be dynamic due to change in user interests over time or their mobility, which affects the number of users served by a small base station, as the users arrive to or depart from the small cell served by the base station. If the file popularity is varying slowly, our proposed solution can be executed whenever such a change occurs to obtain a caching, routing, and link scheduling assignment that maximizes the throughput. On the other hand, if the file popularity is changing quickly, then our algorithm can be executed based on a fixed caching assignment, as explained in Section~\ref{subsec:estimation}, in order to obtain a feasible routing and link scheduling assignment whenever the content popularities change.

\section{Complexity Analysis}\label{sec:performance}
In this section, we analyze the complexity of Algorithm~\ref{alg:alg1} compared to solving the original problem directly. Let $Q$ denote the total number of communication links in the system, then at most $Q = \vert\mathcal{N}\vert \vert\mathcal{K}\vert \vert\mathcal{C}\vert$, and the total number of independent sets is at most $2^Q$. Since there is a variable $f_i$ associated with the $i$-th independent set, the total number of $f_i$ variables is at most $2^Q$. Thus, the total number of variables in the original problem is $2^Q + \vert\mathcal{N}(\vert\mathcal{J}\vert + \vert\mathcal{J}\vert \vert \vert\mathcal{K}\vert) + \vert\mathcal{K}\vert\vert\mathcal{C}\vert$. Thus the complexity of the original problem is $\mathcal{O}(2^Q) = \mathcal{O}(2^{\vert\mathcal{N}\vert \vert\mathcal{K}\vert \vert\mathcal{C}\vert})$.

To analyze the complexity of our algorithm, we consider the complexities of the RMP and PP separately. Note that our algorithm iterates between the RMP and PP, adding a new variable to RMP with each iteration. Initially, the number of variables in RMP is $\mathcal{Q} = Q + \vert\mathcal{N}(\vert\mathcal{J}\vert + \vert\mathcal{J}\vert \vert \vert\mathcal{K}\vert) + \vert\mathcal{K}\vert\vert\mathcal{C}\vert$, since we start with an initial subset of independent sets, where each independent set containts a single communication link. Therefore, in the $l$-th iteration, the total number of variables in RMP is $\mathcal{Q} = Q + l + \vert\mathcal{N}(\vert\mathcal{J}\vert + \vert\mathcal{J}\vert \vert \vert\mathcal{K}\vert) + \vert\mathcal{K}\vert\vert\mathcal{C}\vert$. RMP is a linear program which can be solved using the polynomial interior-point algorithm \cite{ye1991n}, which has a third degree polynomial complexity in terms of the number of variables. If the total number of iterations is $L$, then the complexity of the RMP is $\mathcal{O}(\sum_{l = 1}^{L} \mathcal{Q}^3) = \mathcal{O}(Q^3 + L^4)$.

To analyze the complexity of the PP, note that the total number of variables in the PP is always $Q$, since the PP decides whether to include a communication link in the generated independent set or not. Solving the PP directly using a solver (e.g., CPLEX, which uses a branch and bound algorithm to solve a mixed-integer linear program \cite{cplex}) will result in $\mathcal{O}(2^Q)$ complexity in the worst case. Other methods to solve the PP is to use Sequential Fixing (SF). Sequential Fixing solves a relaxed version of the PP (where the integer variables are relaxed to fractional variables), and then fixes the value of the variables one at a time, and thus the variables can be determined in at most $Q$ iterations of Sequential Fixing. Therefore, the complexity of solving the PP is $\mathcal{O}(L(Q^2)^3)$.

\section{Experimentation}\label{sec:simulation}
\subsection{Setting}
We consider a circular cellular network with a radius of $400m$. The macro base station is located at the center, while the small-cell base stations and the users are independently and uniformly distributed over the network area. There are 10 available secondary channels, each with a bandwidth of $400$kHz, where each small base station can use 5 secondary channels selected randomly and independently of the other small base stations. Each user can also use 5 secondary channels selected randomly and independently of the other users. There is one primary (basic) channel of 1MHz bandwidth that is exclusively used by the macro-cell base station.

The files requested by the users have an average size of 400MB. The popularity of each file is distributed according to a Zipf-distribution with parameter $\zeta = 0.8$\cite{breslau1999web}, where the popularity of a file
of rank $m$ is given as $\frac{1/m^\zeta}{\sum_{j = 1}^{\vert\mathcal{J}\vert}1/j^\zeta}$. Lastly, to stop the algorithm, we set $\epsilon = 0.03$. The simulation parameters are summarized in Table~\ref{tab:tab1}. We compare our system to the Femtocaching system proposed in \cite{golrezaei2012femtocaching}.

\begin{table}[htb]
\caption {Simulation Parameters} \label{tab:tab1}
\begin{center}
    \begin{tabular}{|l|l|}
    \hline
Parameter & Value\\ \hline
Cell Radius & 400 m\\\hline
Primary Channel Bandwidth & 1 MHz\\\hline
Number of Secondary Channels & 10\\\hline
Secondary Channel Bandwidth & 400 kHz\\\hline
Number of Secondary Channels \\available at the SBS & 5\\\hline
Number of Secondary Channels \\available at the users & 5\\\hline
Average File Size & 400 MB\\\hline
File Popularity & Zipf distribution (parameter $\zeta = 0.8$)\\\hline
Average Cache Size & 4 GB (unless stated otherwise)\\\hline
Number of files & 200 (unless stated otherwise)\\\hline
Number of Users & 200 (unless stated otherwise)\\\hline
Number of SBS & 14 (unless stated otherwise)\\\hline
Transmission Range & 100m (unless stated otherwise)\\\hline
Interference Range & Twice the transmission range\\\hline
$\epsilon$ & 0.03\\\hline
\end{tabular}
\end{center}
\end{table}

\subsection{Results}
\subsubsection{Throughput vs. Cache Size}
Figure~\ref{fig:Cache1} shows the results of our simulation experiments measuring the achieved average user rate (Mbps), as the cache size at each SBS is varied. In this simulation, the number of files considered is 200, the total number of users in the cellular network is set to 200, the number of small base stations is set to 14, the transmission range of the small base stations is set to 100m, the interference range is set to twice the transmission range, and all SBSs have the same cache size, which is varied between $0.8$GB to $8$GB. As can be seen, as the cache size increases, the average user throughput increases. This is because the SBS can store more files in its cache and can serve more users without the need to communicate with the macro base station. This will create more opportunities for simultaneous transmissions from different SBSs to satisfy more user requests. Moreover, our system can achieve around 40\% gain over Femtocaching \cite{golrezaei2012femtocaching}, as our system allows for different SBSs to use the same channel simultaneously, given that they are not interfering with each other.

\begin{figure*}
\centering
\begin{minipage}{.28\linewidth}
\includegraphics[width=\linewidth]{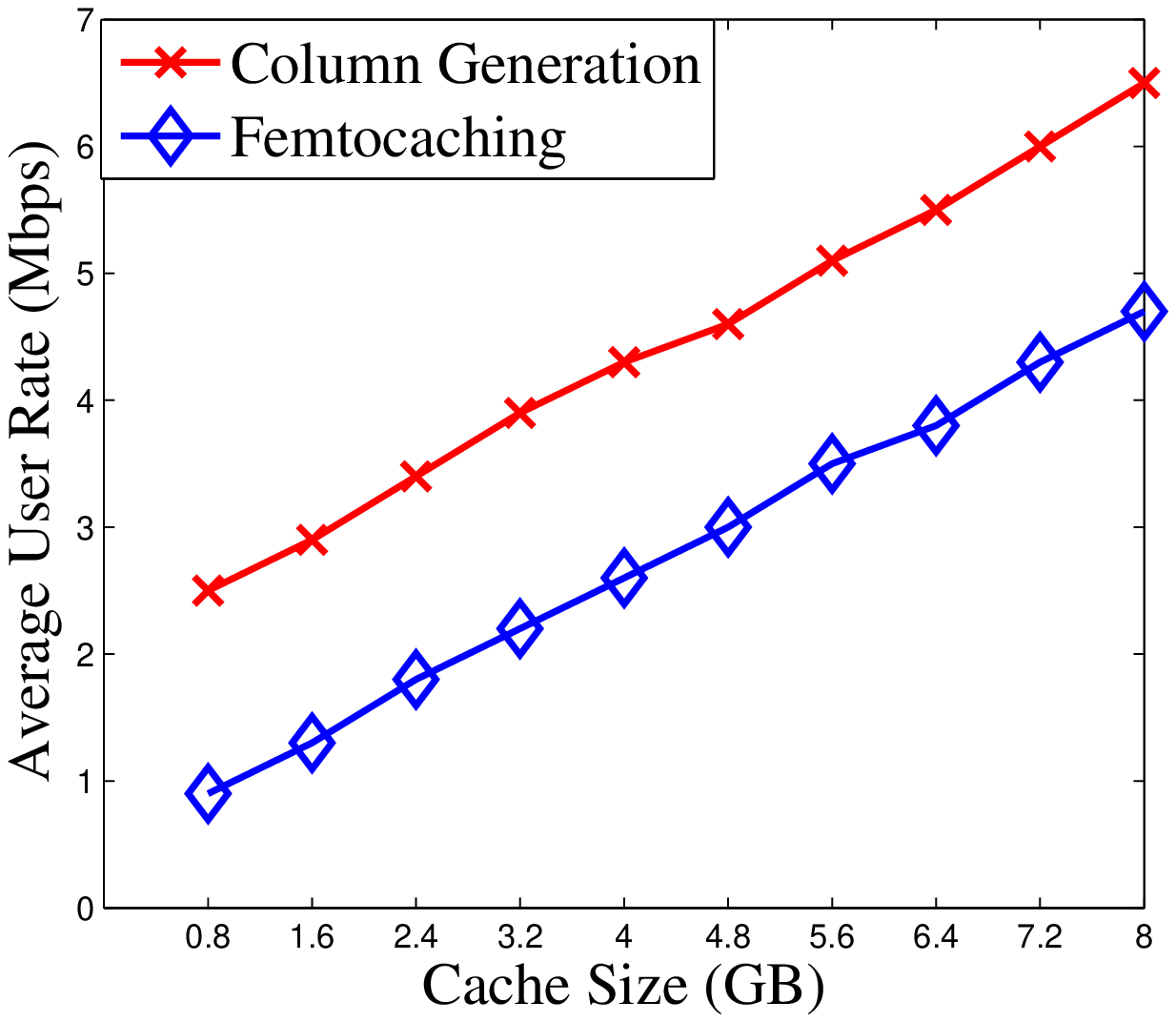}
\captionof{figure}{Throughput vs. cache size (homogeneous case).}
\label{fig:Cache1}
\end{minipage}
\hspace{.05\linewidth}
\begin{minipage}{.28\linewidth}
\centering
\includegraphics[width=\linewidth]{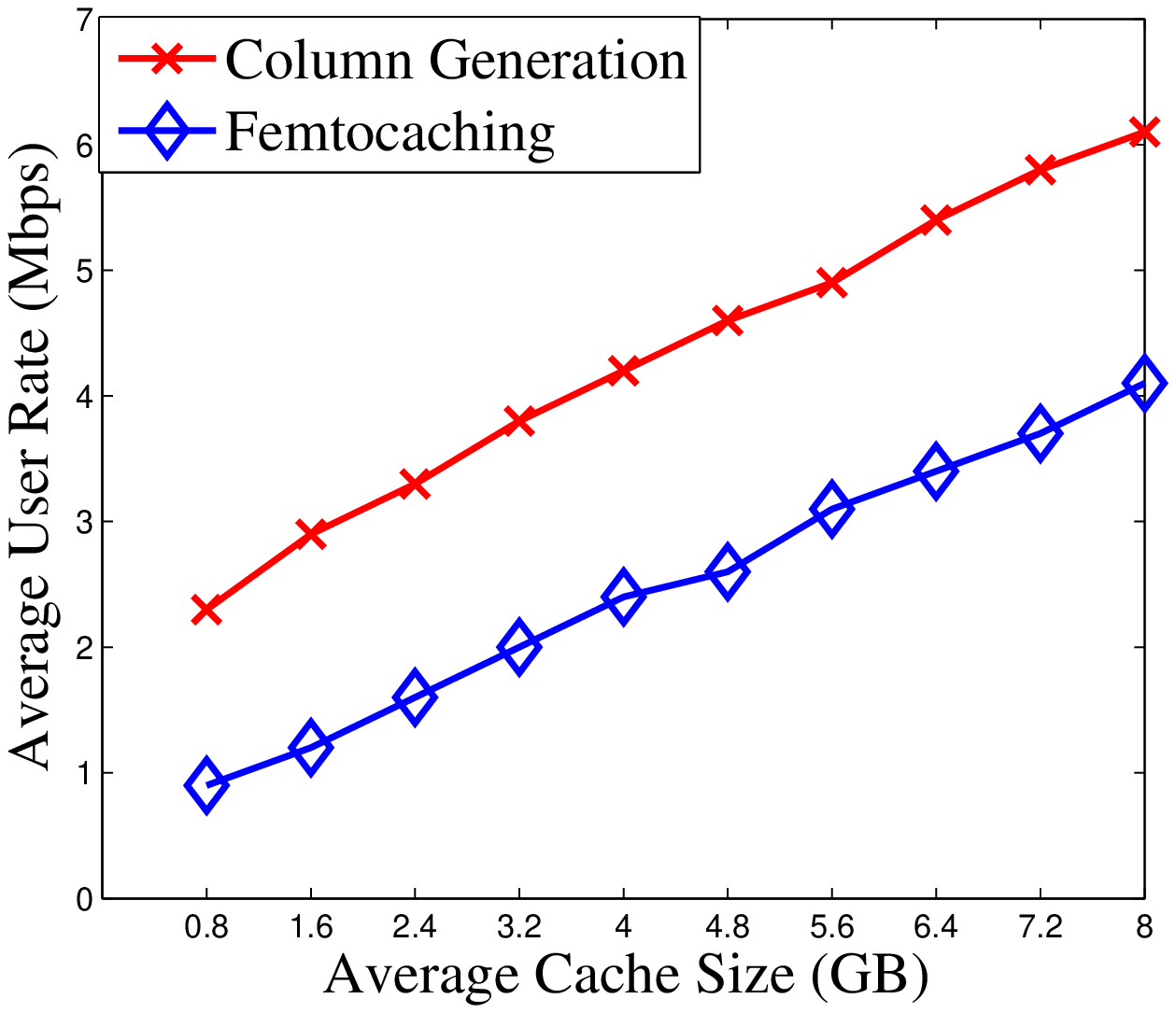}
\captionof{figure}{Throughput vs. cache size (heterogeneous case).}
\label{fig:Cache2}
\end{minipage}
\hspace{.05\linewidth}
\begin{minipage}{.28\linewidth}
\centering
\includegraphics[width=\linewidth]{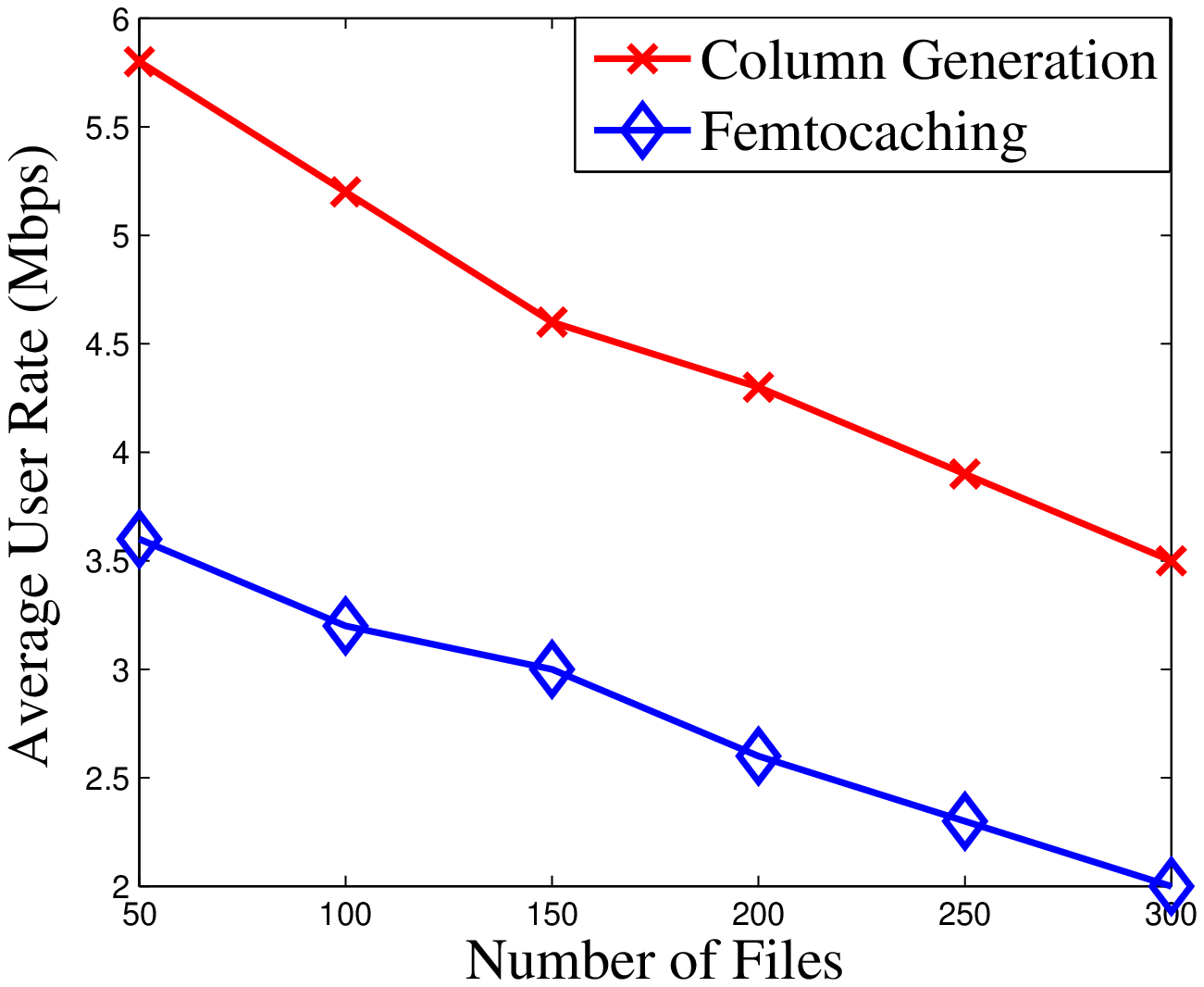}
\captionof{figure}{Throughput vs. number of files.}
\label{fig:files}
\end{minipage}
\end{figure*}

Figure~\ref{fig:Cache2} is similar to Figure~\ref{fig:Cache1} except that the cache size of each SBS is different, but the average cache size is the same as in Figure~\ref{fig:Cache1}. The same trend can be observed in the figure with the possibility of achieving up to 33\% more gain over Femtocaching \cite{golrezaei2012femtocaching} depending on the heterogeneity of the cache size of the SBSs.

\subsubsection{Throughput vs. Number of Files}
In this simulation, the average cache size is set to 4 GB, the total number of users in the cellular network is set to 200, the number of small base stations is set to 14, the transmission range of the small base stations is set to 100m, and the interference range is set to twice the transmission range. Figure~\ref{fig:files} shows the same performance metric, versus the number of files the users can select from. We can see that the achieved average user throughput decreases, as the video file heterogeneity increases. That is because the volume of data transmitted over the same channel necessarily increases in this case, and therefore, it takes increasingly longer to meet all the video file user requests. We also note that our system outperforms Femtocaching \cite{golrezaei2012femtocaching} by 39\% when the number of files is small and by 42\% when the number of files is large. This is due to the opportunities for simultaneous transmission created by our system.

\subsubsection{Throughput vs. Number of Users}
In this simulation, the number of files considered is 200, the average cache size is set to 4 GB, the number of small base stations is set to 14, the transmission range of the small base stations is set to 100m, and the interference range is set to twice the transmission range. Figure~\ref{fig:users} shows the achieved average user throughput, as the number of users is varied. As noted in the figure, as the number of user increases, the average user throughput decreases. This is because when there are more users in the system, there is a higher chance of communication links interfering with each other, which lowers the average user throughput. We also note that our system improves the throughput by 30\% over Femtocaching \cite{golrezaei2012femtocaching} when the number of users is low. That is because in our system, different SBSs can use the same channel for communication, as long as they do not interfere with each other.

\begin{figure*}
\centering
\begin{minipage}{.28\linewidth}
\includegraphics[width=\linewidth]{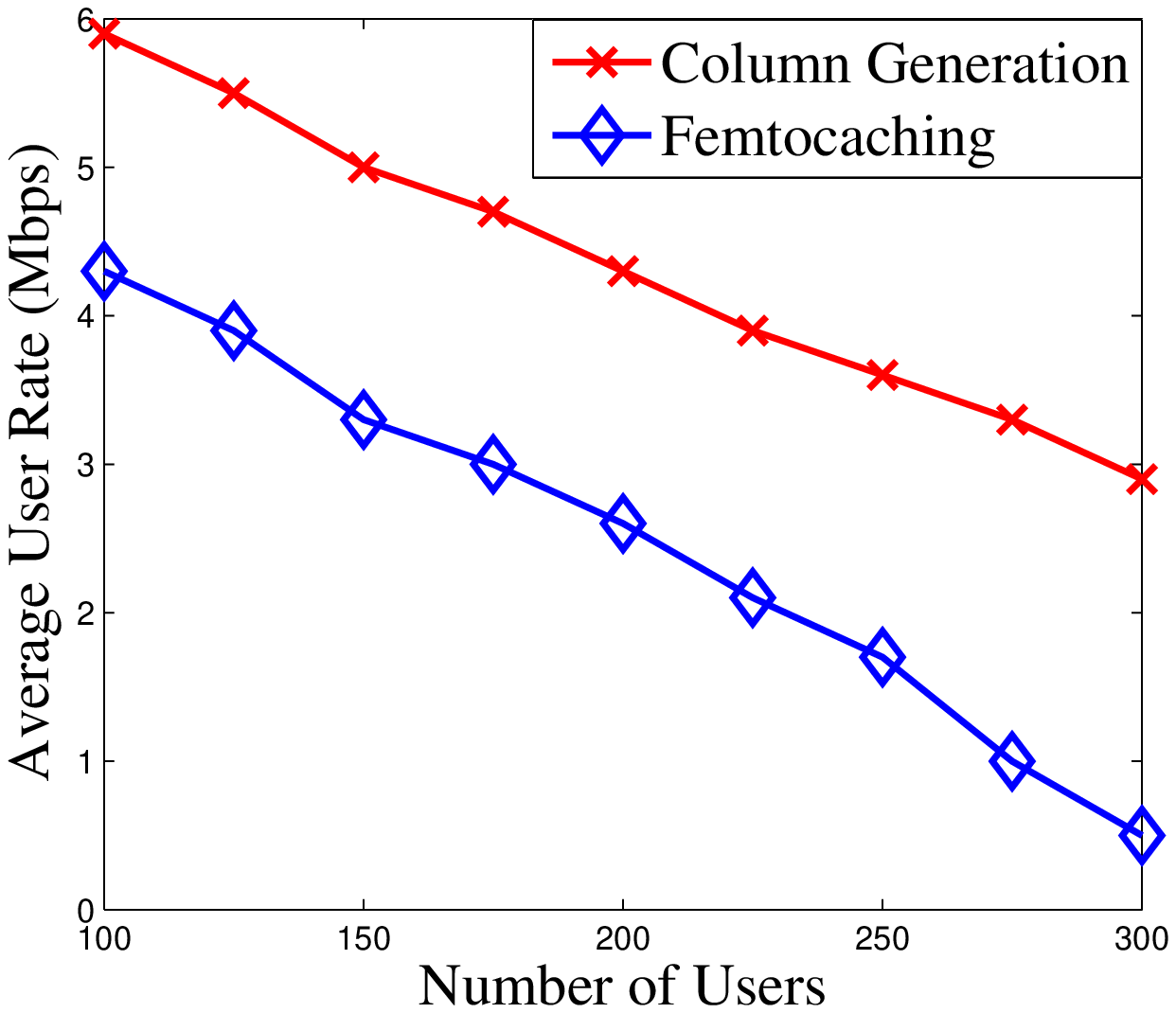}
\captionof{figure}{Throughput vs. number of users.}
\label{fig:users}
\end{minipage}
\hspace{.05\linewidth}
\begin{minipage}{.28\linewidth}
\centering
\includegraphics[width=\linewidth]{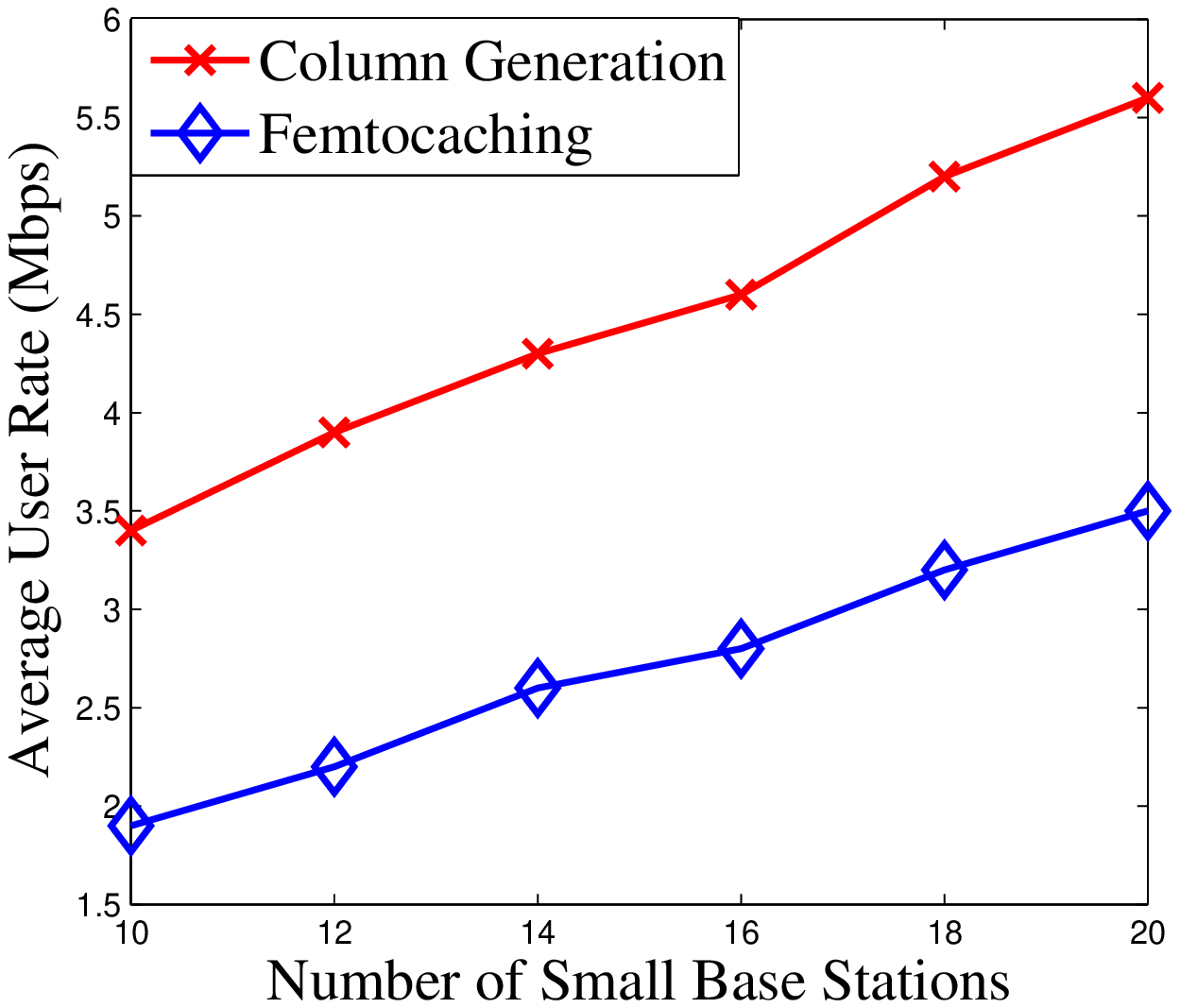}
\captionof{figure}{Throughput vs. number of small base stations.}
\label{fig:sbs}
\end{minipage}
\hspace{.05\linewidth}
\begin{minipage}{.28\linewidth}
\centering
\includegraphics[width=\linewidth]{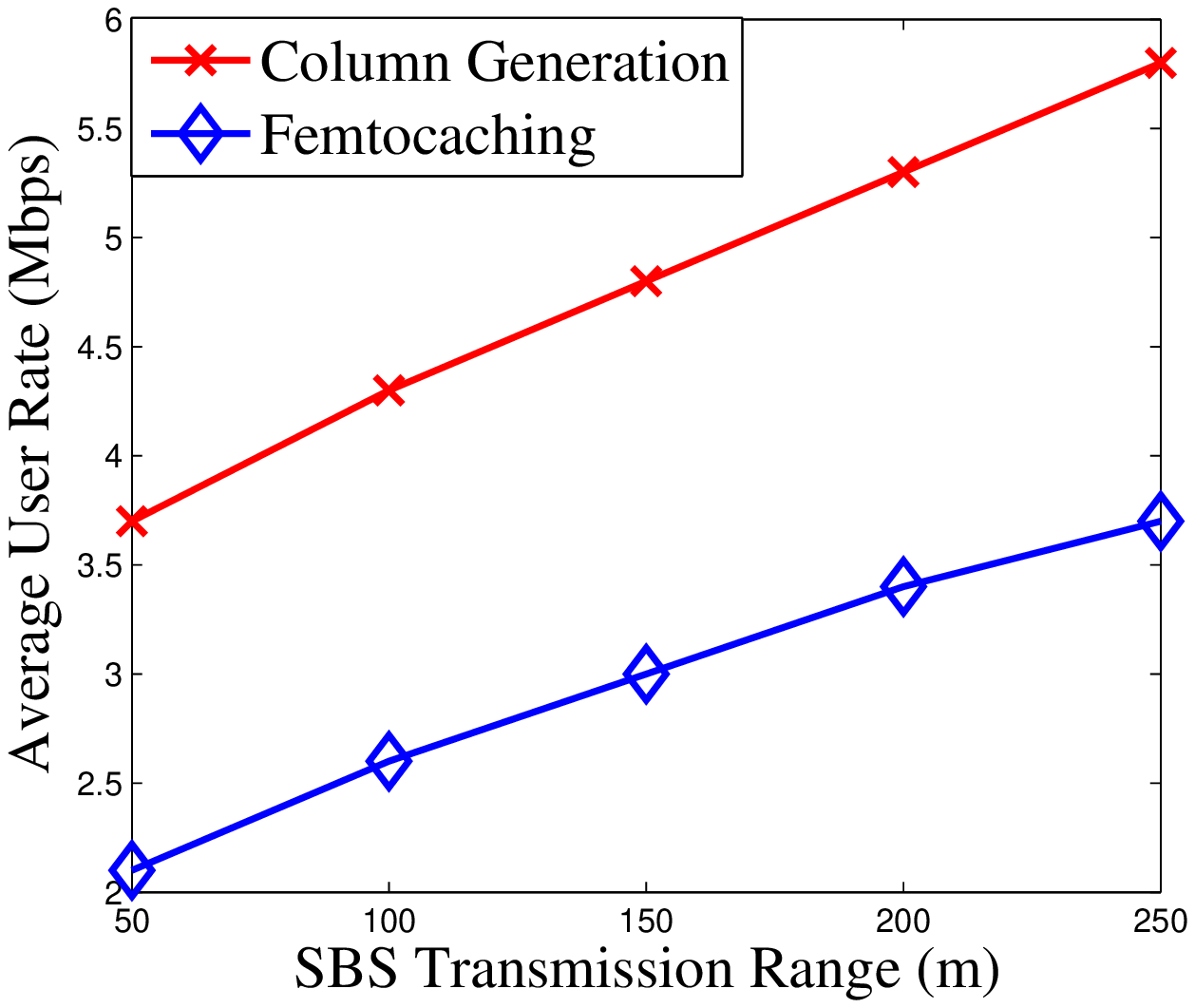}
\captionof{figure}{Throughput vs. Radius of Small Base Stations.}
\label{fig:radius}
\end{minipage}
\end{figure*}

\subsubsection{Throughput vs. Number of Small Base Stations}
In this simulation, the number of files considered is 200, the average cache size is set to 4 GB, the total number of users in the cellular network is set to 200, the transmission range of the small base stations is set to 100m, and the interference range is set to twice the transmission range. Figure~\ref{fig:sbs} shows the results of the simulation when the number of small base stations is varied. As can be seen from the figure, as the number of small base stations increases, the average user throughput increases. That is because each small-cell base station will serve a lower number of users in this case, and this will create more opportunities for simultaneous transmissions. Moreover, as the number of small base stations increases, there is a higher chance that a file is served by a small base station instead of the macro base station, since serving a request from the macro base station on a certain channel will forbid everyone else from using that same channel. As noted from Figure~\ref{fig:sbs}, relative to Femtocaching, our system improves the throughput by 40-42\%, as it is able to exploit the same channel for simultaneous transmissions at multiple small-cells, serving different user requests at the same time.

\subsubsection{Throughput vs. Radius of Small Base Stations}
In this simulation, the number of files considered is 200, the average cache size is set to 4 GB, the number of small base stations is set to 14, and the total number of users in the cellular network is set to 200. Figure~\ref{fig:radius} examines the impact of the transmission range of the small base stations. The interference range is set to twice the transmission range. We can see that the average throughput of the users increases as the transmission range increases. This is because an increasing number of users fall within the transmission range of the small base stations thereby and can be served by them instead of the macro base station, as any channel used by the macro base station for transmission will not be used by the small base station, at the same time. Moreover, Figure~\ref{fig:radius} shows that our proposed solution achieves a 34-46\% additional throughput gain over Femtocaching, since our solution can exploit simultaneous transmissions using the same channel at multiple small base stations.


\section{Video Application Quality Implications}
The average user transmission rate gains enabled by our system over the state-of-the-art will map to equivalent gains in video application quality that in turn will enhance the user experience considerably. Concretely, the 34-46\% higher data rate, demonstrated in our experiments, can support the delivery of video content featuring 2-4 dB higher video quality (Y-PSNR), for the typical online video spatial resolutions delivered today \cite{Chakareski:15}. Alternatively, the higher data rate can support the delivery of higher resolution videos, featuring the same Y-PSNR video quality, again enhancing the user experience. Increasing the temporal resolution (frame rate) of the delivered video content can be yet another benefit supported by the higher network throughput enabled by our optimization framework. Finally, the latter can also be utilized to reduce the start-up (play-out) delay of the video application, by delivering the video content faster to the user. As noted in user studies, this factor is the most-critical performance metrics affecting the quality of experience of the users \cite{MatosCMSCK:12}.

We anticipate that further gains in performance can be achieved if the content is represented in a scalable coding format, e.g., SVC \cite{ITU:05,SchwarzMW:07}, as further trade-offs between caching space and user video quality can be explored thereby. This is, however, beyond the scope of the present paper and can represent a prospectively fruitful follow-up investigation. 

\section{Conclusion}\label{sec:conclusion}
We studied the problem of joint caching, routing, and channel assignment for video delivery over coordinated multi-cell systems of the future Internet. We formulated a novel optimization framework based on the column generation method that is used to solve large-scale linear programming problems. To control the complexity of the optimization, we have formulated it such that it features a restricted master subproblem and a pricing subproblem. The former uses only a subset of the original variables, while the latter is used to determine whether another original variable should be introduced into the master problem formulation, if that leads to lower objective function value. Our formulation integrates the effect of interference by accounting for the conflict graph between the wireless transmission links comprising the system. To further control computational complexity, we have designed a $(1\pm\epsilon)$ approximation algorithm that computes the optimal solution at lower complexity, and we have proved its approximation bounds. Our simulation experiments demonstrate that our system can deliver consistent 34-46\% gains in network throughput over the state-of-the-art Femtocaching system, over a wide range of system parameters and network conditions. These advances will result in equivalent gains in video application quality that in turn will enhance the quality of experience of the user.

\bibliographystyle{IEEEtran}
\bibliography{biblo2}

\begin{biography}[{\includegraphics[width=1in,
height=1.25in,clip,keepaspectratio]{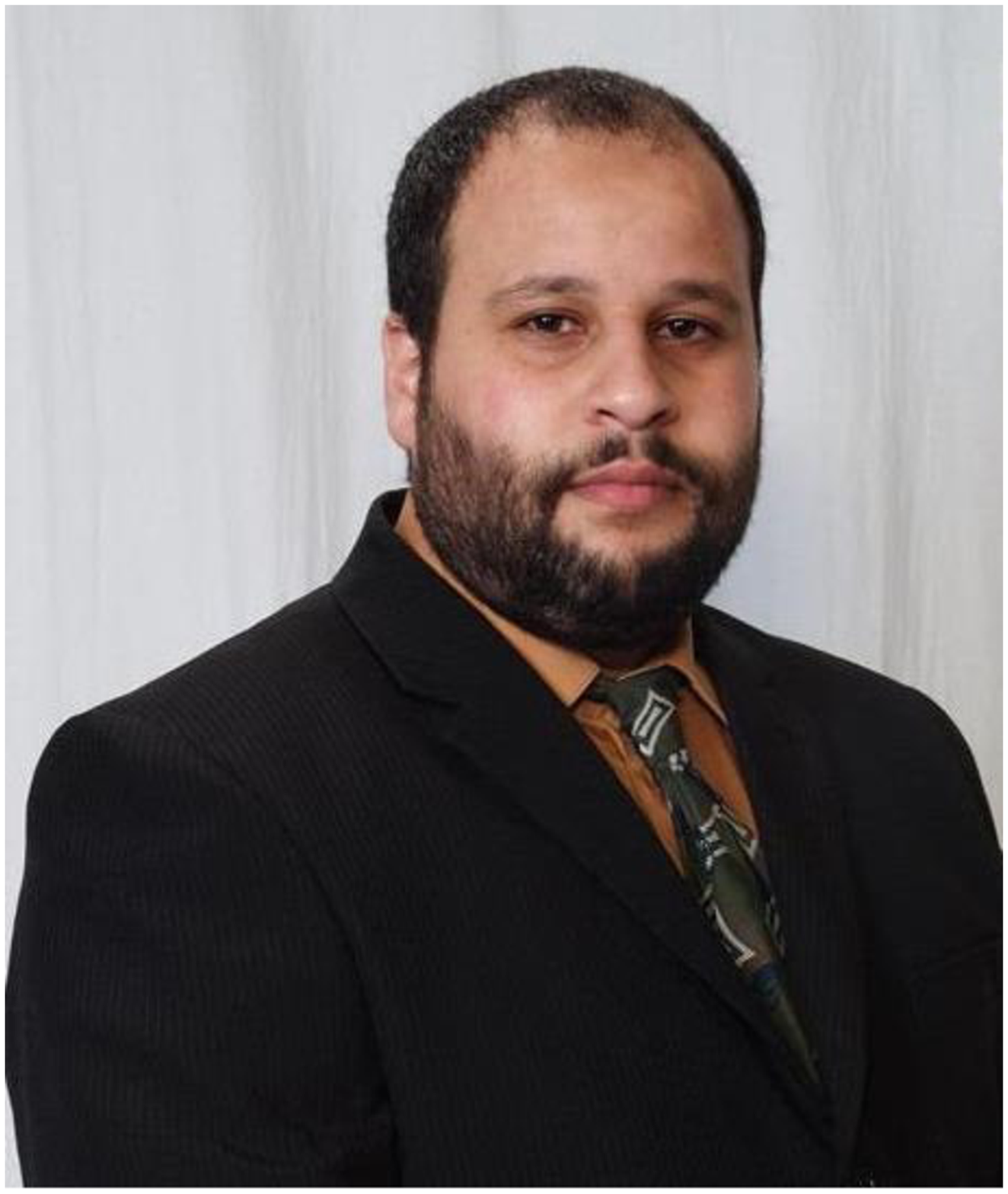}}]{Abdallah Khreishah} is an assistant professor in the Department of Electrical and Computer Engineering at New Jersey Institute of Technology. His research interests fall in the areas of visible-light communication, green networking, network coding, wireless networks, and network security. Dr. Khreishah received his BS degree in computer engineering from Jordan University of Science and Technology in 2004, and his MS and PhD degrees in electrical \& computer engineering from Purdue University in 2006 and 2010. While pursuing his PhD studies, he worked with NEESCOM. He is a Member of the IEEE and the chair of North Jersey IEEE EMBS chapter.
\end{biography}

\begin{biography}[{\includegraphics[width=1in,
height=1.25in,clip,keepaspectratio]{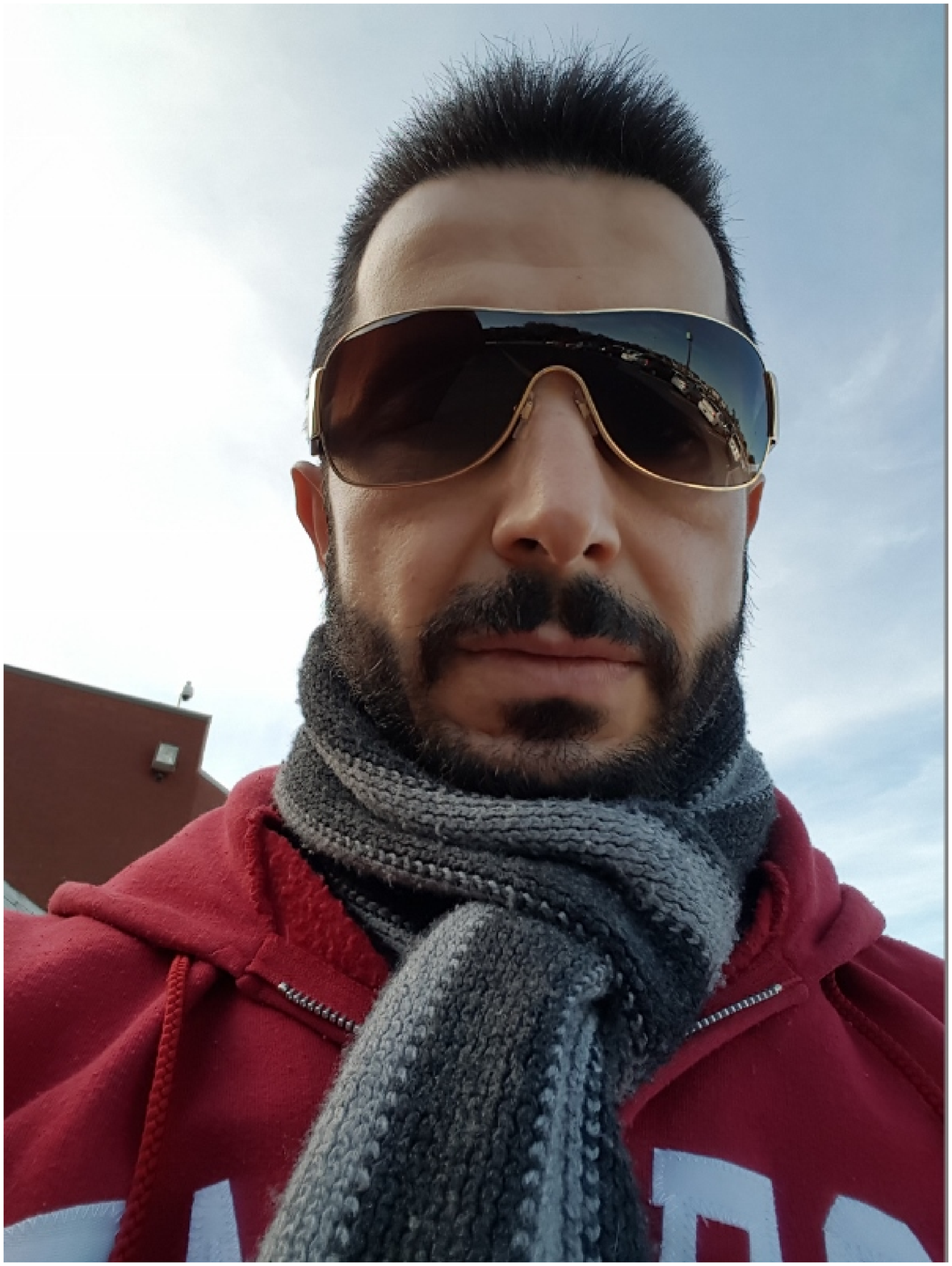}}]{Jacob Chakareski} completed the M.Sc. and Ph.D. degrees in electrical and computer engineering at the Worcester Polytechnic Institute (WPI), Worcester, MA, USA, Rice University, Houston, TX, USA, and Stanford University, Stanford, CA, USA. He is an Assistant Professor of Electrical and Computer
Engineering at the University of Alabama. He was a Senior Scientist at Ecole Polytechnique Federale de Lausanne (EPFL), Lausanne, Switzerland, where he conducted research, supervised students, and lectured. He also held research positions with Microsoft, Hewlett-Packard, and Vidyo, a leading provider of Internet telepresence solutions. Chakareski has authored one monograph, three book chapters, and over 130 international publications, and holds 5 US patents. His current research interests include immersive visual communication, future Internet architectures, graph-based signal and information processing, and social computing. He also pursues ultrasound and interactive 3D video applications in telemedicine, remote sensing, biomedicine, and community-based health care.

Dr. Chakareski is a member of Tau Beta Pi and Eta Kapa Nu. He is a recipient of the Technical University Munich Mobility Fellowship, the University of Edinburgh Chancellor’s Fellowship, and fellowships from the Soros Foundation and the Macedonian Ministry of Science. He received the Texas Instruments Graduate Research Fellowship, the Swiss NSF Ambizione Career Development Award, the Best Student Paper Award at the SPIE VCIP 2004 Conference, and the Best Paper Award of the Stanford Electrical Engineering and Computer Science Research Journal for 2004. He was the Publicity Chair of the Packet Video Workshop 2007 and 2009, and the Workshop on Emerging Technologies in Multimedia Communications and Networking at ICME 2009. He has organized and chaired a special session on telemedicine at MMSP 2009. He was the Technical Program Co-Chair of Packet Video 2012 and the General Co-Chair of the IEEE SPS Seasonal School on Social Media Processing 2012. He was a Guest Editor of the Springer PPNA Journal’s 2013 special issue on P2P-Cloud Systems. He is the Technical Program chair of Packet Video 2016 and the guest editor of the IEEE TCSVT special issue on Mobile Visual Cloud (June 2016). Chakareski is an Advisory Board member of Frame, an innovative cloud computing start-up with a bright future. He is an IEEE Senior Member. For further information, please visit www.jakov.org.
\end{biography}

\begin{biography}[{\includegraphics[width=1in,
height=1.25in,clip,keepaspectratio]{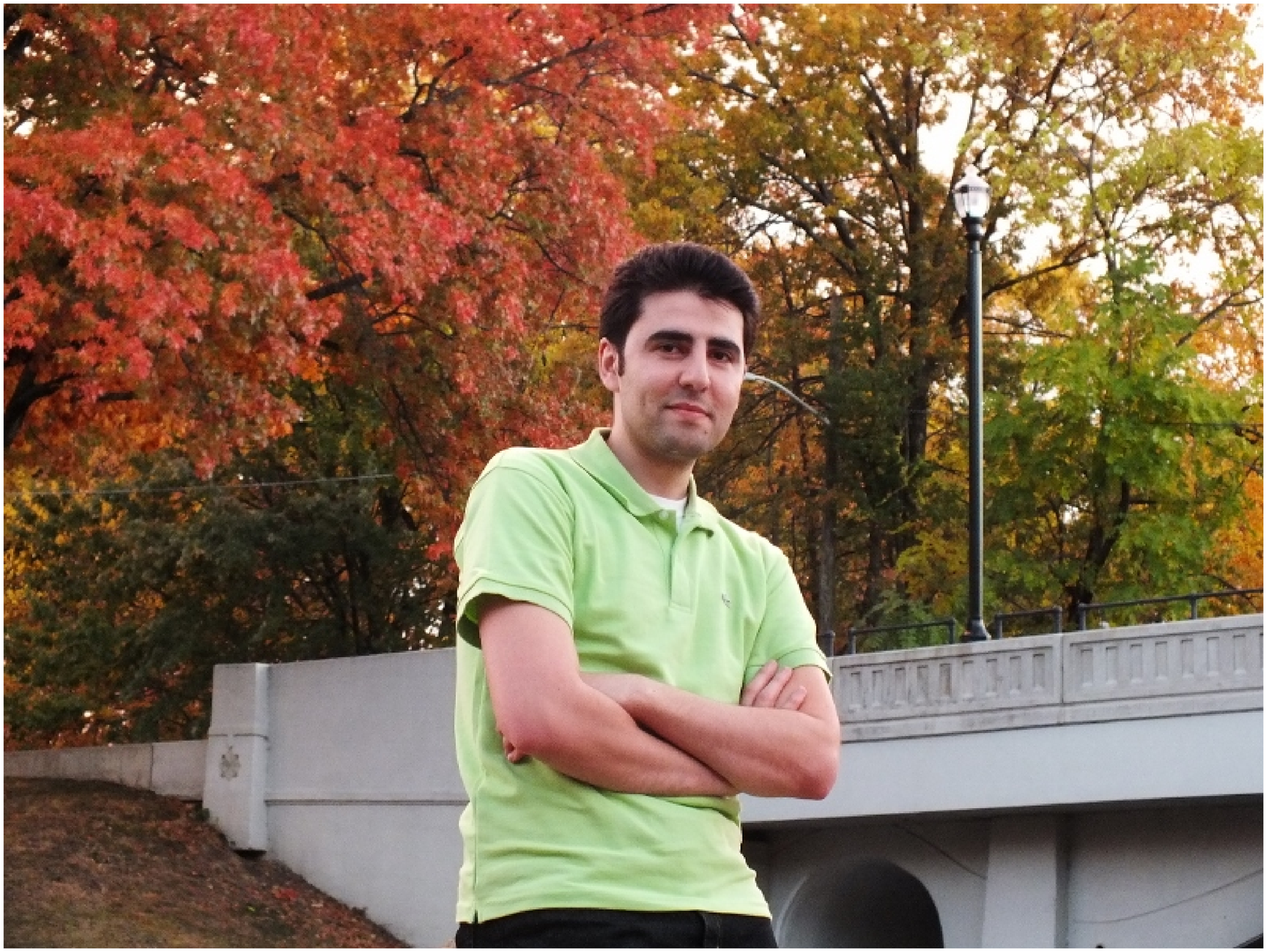}}]{Ammar Gharaibeh} is a PhD student at the ECE department of New Jersey Institute of Technology. He received his M.S. degrees in Computer Engineering from Texas A\& M University in 2009. Prior to that, he received his B.S. degree with honors from Jordan University of Science and Technology in 2006. His research interests spans the areas of wireless networks and network caching.
\end{biography}

\end{document}